\documentclass[11pt]{article} %

\usepackage{amsmath,amsthm,amssymb, tabu} %
\usepackage{fullpage}
\usepackage{booktabs} %
\usepackage{dsfont} %
\usepackage{enumitem} %
\usepackage{float} %
\usepackage[margin=1in]{geometry} %
\usepackage{graphicx} %
\usepackage{mathtools}
\usepackage{enumitem}
\usepackage{verbatim}

\usepackage[utf8]{inputenc} %
\usepackage{mdframed} %
\usepackage{mathpazo} 
\usepackage[T1]{fontenc}
\usepackage{soul} %
\usepackage{suffix} 
\usepackage{tikz} %
\usepackage{xcolor} %
\usepackage{xspace} %
\usepackage[scr=rsfso]{mathalfa} 
\usepackage{algorithm2e}

\definecolor{weborange}{rgb}{.8,.3,.3}
\definecolor{webblue}{rgb}{0,0,.8}
\definecolor{internallinkcolor}{rgb}{0,.5,0}
\definecolor{externallinkcolor}{rgb}{0,0,.5}

\usepackage{cleveref}

\usetikzlibrary{calc,positioning}

\mdfdefinestyle{figstyle}{ %
  linecolor=black!7, %
  backgroundcolor=black!7, %
  innertopmargin=10pt, %
  innerleftmargin=25pt, %
  innerrightmargin=25pt, %
  innerbottommargin=10pt %
}

\definecolor{White}{rgb}{1,1,1} %
\definecolor{Black}{rgb}{0,0,0} %
\definecolor{LightGray}{rgb}{.8,.8,.8} %
\colorlet{ChannelColor}{LightGray} %
\colorlet{ChannelTextColor}{Black} %
\colorlet{ReadoutColor}{White} %

\numberwithin{equation}{section}
\newtheorem{theorem}{Theorem}
\newtheorem{lemma}[theorem]{Lemma}
\newtheorem{algo}[theorem]{Algorithm}
\newtheorem{corollary}[theorem]{Corollary}

\newtheorem{definition}[theorem]{Definition}

\let\epsilon=\varepsilon %

\newcommand{\microspace}{\mspace{.5mu}} %
\newcommand{\ket}[1]{\ensuremath{\lvert\microspace #1
    \microspace\rangle}} %
\newcommand{\bra}[1]{\ensuremath{\langle\microspace #1
    \microspace\rvert}} %
\newcommand{\ketbra}[2]{\ensuremath{\lvert\microspace #1
    \microspace\rangle\! \langle \microspace #2 \microspace \rvert}} %

\newcommand{\class}[1]{\mathsf{#1}} %

\newcommand{\NP}{\class{NP}} %
\newcommand{\MA}{\class{MA}} %
\newcommand{\QMA}{\class{QMA}} %
\newcommand{\QMIP}{\class{QMIP}} %
\WithSuffix\newcommand\QMIP*{\ensuremath{\class{QMIP}^*}} %
\newcommand{\PSPACE}{\class{PSPACE}} %
\newcommand{\revPSPACE}{\class{revPSPACE}} %
\newcommand{\MIP}{\class{MIP}} %
\WithSuffix\newcommand\MIP*{\ensuremath{\class{MIP}^*}} %
\newcommand{\QMAc}{\ensuremath{\class{QMA}^{c=1}}} %
\newcommand{\QMAexp}{\ensuremath{\class{PreciseQMA}}} %
\newcommand{\QMAexpc}{\ensuremath{\class{PreciseQMA}^{c=1}}} %
\newcommand{\MAexp}{\ensuremath{\class{MA}_{\class{exp}}}} %
\newcommand{\NPpp}{\ensuremath{\class{NP}^{\class{PP}}}} %
\newcommand{\poly}{\mathrm{poly}}

\begin{document}

\title{A Simple Proof of $\QMAexp = \PSPACE$}
\author{Yulong Li \\
Columbia University}
\date{} 
\maketitle

\begin{abstract}
   We give an alternative proof of $\QMAexp = \PSPACE$,  first proved by Fefferman and Lin (\emph{Innov. Theor. Comp. Sci. 2018}), where $\QMAexp$ is the class Quantum Merlin-Arthur with inverse exponential completeness-soundness gap. We adapt the proof of Quantum Cook-Levin Theorem to prove the inclusion $\PSPACE \subseteq \QMAexp$. 
\end{abstract}

\section{Introduction}

A crucial focus of quantum complexity theory is to study quantum counterparts of classical complexity classes. By comparing the classical and quantum complexity classes, we may better understand the power of quantum computation models. One particularly important quantum complexity class is Quantum Merlin-Arthur ($\QMA$). It is the quantum analog of Merlin-Arthur ($\MA$), and it can also be seen as the quantum equivalent of $\NP$ because of the close relationship between the canonical $\NP$-complete problem, the satisfiability problem, and the natural $\QMA$-complete problem, the local Hamiltonian problem. The formal definition of $\QMA$ is as following.

\begin{definition}
A promise problem $L = (L_{yes}, L_{no})$ is in $\QMA(c, s)$ if there exists a uniform family of quantum circuits $\{V_x \}_x \in \{0,1\}^{|x|}$, each of size at most $poly(|x|)$ and acting on $poly(|x|)$ qubits, such that:
For $x \in L_{yes}$, there exists a $poly(|x|)$-qubit state $\ket\psi$ such that $\Pr[V_x \text{ accepts $\ket{\psi}$}] \geq c$; for $x \in L_{no}$, for all $poly(|x|)$-qubit state $\ket\psi$:
$\Pr[V_x \text{ accepts $\ket{\psi}$}] \leq s$. The parameter $c$ and $s$ are known as the completeness and soundness of the problem, respectively, and the quantity $c-s$ is known as the \emph{completeness-soundness gap}, or \emph{gap} for short. $\QMA$ is typically defined as $\QMA(c,s)$ for all $c - s$ that is inverse polynomial in $n$.
\end{definition}

One interesting question to ask is how will the precision of values of $c$ and $s$ change the power of $\QMA$. Kitaev showed that as long as the gap $c-s$ is at least inverse polynomial, we can amplify the gap to a constant \cite{KSV02}. It's then natural to ask what if the promise gap is much smaller. Surprisingly, Fefferman and Lin showed that $\QMA$ class with inverse exponentially small gap, denoted as $\QMAexp$, is equal to $\PSPACE$ \cite{FL16}.


\begin{definition}
$\QMAexp = \cup_{c, s \in [0,1], c-s \geq 2^{-poly(n)}}\QMA(c, s)$ where $n$ is the length of the input.
\end{definition}

\begin{theorem} [~\cite{FL16}]
\label{QMA_in_PSPAE}
$\QMAexp = \PSPACE$ 
\end{theorem}

On the other hand, we can upper bound $\MAexp$, the classical analogy of $\QMAexp$ with $\NPpp$, which is widely believed to be strictly contained in $\PSPACE$. This implies a possible separation between the classical complexity class $\MA$ and the quantum analogue $\QMA$, at least in the inverse exponential gap setting.



We give an alternate, arguably simpler, proof that $\QMAexp = \PSPACE$, specifically the direction $\QMAexp \supseteq \PSPACE$. The direction $\QMAexp \subseteq \PSPACE$ is fairly intuitive. Given a polynomial size proof state and polynomial time verification algorithm, we can simulate the protocol in polynomial space by guessing the proof and performing amplification even when the gap is inverse exponential \cite{FL16}. The reverse direction, $\QMAexp \supseteq \PSPACE$, is more involved. In Fefferman and Lin's original proof, they define the problem \textit{Gapped Succinct Matrix Singularity (GSMS)}, which, informally speaking, takes in a sparse matrix whose smallest eigenvalue is promised to be either $0$ or at least inverse exponential and decides between the two cases. They gave a $\QMAexp$ protocol for $GSMS$ based on Hamiltonian simulation and phase estimation and showed that $GSMS$ is $\PSPACE$-hard by encoding the configuration graph  of a $\PSPACE$ machine into a sparse matrix that satisfies the promise gap of $GSMS$. In a follow-up paper by Fefferman and Lin, they prove a more general result about unitary quantum space, of which $\QMAexp = \PSPACE$ is a corollary. But their main result still requires a nontrivial reduction to a general linear algebra problem \cite{fefferman2016complete}.

The proof inspired us to encode the computation of a $\PSPACE$ machine \emph{directly} using the Feynman-Kitaev circuit-to-Hamiltonian construction \cite{KSV02}. This proof technique was used to show the QMA-completeness of the Local Hamiltonian problem. In our opinion, the direct encoding of the computation history and the close relationship between Hamiltonian problems and $\QMA$ make the proof $\PSPACE \subseteq \QMAexp$ more intuitive and sidestep the linear algebraic problem of \textit{GSMS}. As an easy corollary of the proof, we show that $\QMAexp$ with perfect completeness, denoted as $\QMAexpc$, is equal to $\QMAexp$ (Corollary \ref{qmac=1}), which is first shown in \cite{fefferman2016complete}. Note that perfect completeness does not change the class $\MA$ \cite{zachos1987probabilistic}, the classical analogue of $\QMA$, but it is an open problem whether $\QMA$ is equal to $\QMAc$. In fact, Aaronson showed that $QMA \neq \QMAc$ relative to a quantum oracle \cite{S08}.

\section{Reversible space computation}
One crucial difference between classical and quantum computation is that for quantum circuits, we often require the operations to be unitaries which are inherently reversible, while we typically do not have reversibility constraints for Turing machines. Thus, we need the following theorem about reversible classical computation.
\begin{definition}
We say a promise problem $L = (L_{yes}, L_{no})$ is in $\revPSPACE$ (reversible polynomial space) if it can be decided by a polynomial space reversible Turing machine, i.e., a machine for which every conﬁguration has at most one immediate predecessor. 
\end{definition}

\begin{theorem} [~\cite{B89}]
\label{revpspace}
$\PSPACE = \revPSPACE$
\end{theorem}

Theorem 5 shows that without loss of generality a polynomial-space computation can be made reversible, and thus can be simulated by a unitary quantum circuit with perfect completeness and soundness.

\begin{corollary}
\label{cor}
If $L \in \PSPACE$, $L$ is recognized by a family of unitary quantum circuits $\{C_n\}$ which run in exponential time and polynomial space in $n$; moreover,  if $x \in L$, then the circuit accepts with probability $1$, otherwise it accepts with probability $0$.
\end{corollary}

\section{QMA with exponentially small gap}

Our first step of the proof is to construct a Hamiltonian with inverse exponential promise gap given a $\PSPACE$ machine (or equivalently a quantum circuit using exponential time and polynomial space) based on history state construction, which is shown in Lemma. \ref{pspace_ham}. The second step is to give a $\QMAexp$ protocol for the Local Hamiltonian Problem with inverse exponential promise gaps, which is Algorithm \ref{qma_protocal}. 


\begin{lemma}
\label{pspace_ham}
If $L$ is recognized with perfect completeness and soundness by a family of quantum circuits $\{C_n\}$ which run in exponential time and polynomial space in $n$, then for all $x \in \{0,1\}^n$, there exists a Hamiltonian $H_x$ such that, if $x \in L$, $\min_{\psi} \bra \psi H_x \ket \psi = 0$, and if $x \notin L$, $\min_{\psi} \bra \psi H_x \ket \psi \geq \frac{1}{2(T+1)^3}$.
\end{lemma}

\textbf{Proof:}
We adapt the Quantum Cook-Levin Theorem to the setting of exponentially-long quantum computations. The idea is to construct a Hamiltonian to check to check whether a state encodes a valid history state of a circuit.

\textbf{Construction:}

Suppose the quantum circuit $C_n$ that acts on $S(n)$ qubits and consists of $T(n) = \exp^{poly(n)}$ unitary two-qubit gates, $U_1, \cdots, U_T$.  For $b \in \{0,1\}, i \in [S(n)]$, we define $\Pi_i^{\ket {b}}$ to be the projection operator that operates on qubit number $i$ and projects on the subspace spanned by $\ket{b}$.
 
$$H_{in}^i = \Pi_i^{\ket {\lnot x_i}} \otimes \ketbra{0}{0}_C$$
$$H_{out} = \Pi_1^{\ket 0} \otimes \ketbra{T}{T}_C$$
$$H_{prop}^t = \frac{1}{2} (I \otimes \ketbra{t}{t} + I \otimes \ketbra{t-1}{t-1} - U_t \otimes \ketbra{t}{t-1} - U_t^\dagger \otimes \ketbra{t-1}{t})$$
$$H_x = \sum_{i=1}^n H_{in}^i + H_{out} + \sum_{t=1}^{T} H_{prop}^t$$

for $t \in \{0, 1, \dots, T\}$, where $\ket t_C$ is the binary representation of $t$ taking $\log T$ qubits. If $x \in L$, we have the history of computation $\ket \eta = \frac{1}{T+1}\sum_{i=1}^T U_t\cdots U_{1} \ket \chi \otimes \ket t$ where $\ket \chi = \ket x \otimes \ket 0 ^ {\otimes S-n}$ is the start state encoding the work space including an input $x$. It is straightforward to check that $\bra \eta H \ket \eta = 0$ if $x \in L$ using the fact that $C_n$ accepts $x$ with probability $1$ if $x \in L$. If $x \notin L$, following the analysis in ~\cite{aharonov2002quantum}, one can show that $\bra \psi H \ket \psi  \geq \frac{1}{2(T+1)^3}$ for all $\ket \psi$.

\begin{algo}
\label{qma_protocal}

We give a $\QMAexp$ protocol, which given a polynomial size proof state, check the eigenvalue of $H_x$ defined in \ref{pspace_ham} in polynomial time such that if $\min_{\ket \psi} \bra \psi H_x \ket \psi = 0$, always accepts; otherwise, reject with at least inverse exponential probability. 

\begin{enumerate}
    \item Pick $y \in [T+n+2]$ uniformly at random.
    \item Define
    \[
      H_{test}(y) = 
      \begin{cases}
        H_{prop}^{y-1}, & y \in  [T+1] \\
        H_{in}^{y-T-1}, & y \in [T+2, T+n+1] \\
        H_{out}, & y = T+n+2
      \end{cases}
    \]
    \item If $H_{test}(y) = H_{in}^i$ or $H_{test} = H_{out}$, which are projections onto standard basis, measure $\ket \psi$, reject if the measurement is in the projected space.
    \item If  $H_{test}(y) = H_{prop}^t$, we apply several transformations. Note that 
    $$R_t^\dagger H_{prop}^t R_t = \frac{1}{2} I \otimes (\ket t - \ket {t-1}) (\bra t - \bra {t-1})$$ where $R =\sum_{t=0}^T U_t U_{t-1}\cdots U_1\otimes \ketbra{t}{t}$. Thus, we can design a procedure that involves controlled unitaries of $U_t$ and $U_{t+1}$ that results in measurement of $1$ with probability $\bra \psi H_{prop}^t\ket \psi$ in which case we reject.
    
\end{enumerate}
The total probability of rejection is
$$\frac{1}{T+n+2}\sum_{y=1}^{T+n+2} \bra \psi H_{test}(y)\ket \psi = \frac{1}{T+n+2} \bra \psi H_x \ket \psi$$
For $x \in L$, this equals to $0$; for $x \notin L$, this is at least $\exp(-\poly(n))$.
\end{algo}

Finally, the result immediately implies that we can always boost the completeness of a $\QMAexp$ protocol to $1$, which is first shown in \cite{fefferman2016complete}.

\begin{definition}
$\QMAexpc = \QMA_{poly}(1, 1 - 2^{-poly})$
\end{definition}

\begin{corollary}
\label{qmac=1}
\QMAexp = \QMAexpc
\end{corollary}
\begin{proof}
For any $L \in \QMAexp$, by Theorem \ref{QMA_in_PSPAE} and Corollary \ref{cor}, we have a uniform family of unitary quantum circuits $\{C_n\}$ which run in exponential time and polynomial space in $n$ that accepts a string in $L$ with probability $1$ and rejects otherwise. Then by Lemma \ref{pspace_ham} and Algorithm \ref{qma_protocal}, we have a $\QMAexpc$ protocol for $L$.
\end{proof}

However, this is an indirect reduction that goes through $\PSPACE$. It is thus interesting to ask whether there is a direct reduction from $\QMAexp$ to $\QMAexpc$.

\section{Acknowledgements}

We are very thankful to Henry Yuen for introducing the problem and his guidance throughout the process. We are also thankful for helpful discussions with Adrian She at various stages of the work.

\bibliographystyle{alpha}
\bibliography{bibliography.bib}

\end{document}